\documentclass[runningheads]{llncs}
\usepackage[T1]{fontenc}
\usepackage{graphicx}
\usepackage{listings}

\usepackage[utf8]{inputenc}    
\usepackage[T1]{fontenc}       
\usepackage{lipsum}            
\usepackage{amssymb}  
\usepackage{graphicx}          
\usepackage{hyperref}          
\usepackage{comment}

\usepackage{tikz}
\usetikzlibrary{matrix}
\usetikzlibrary{positioning}

\usepackage{xspace}
\usepackage{bm}

\usepackage{amsmath}
\usepackage{graphicx}
\usepackage{algorithm}
\usepackage{algpseudocode}
\usepackage{url}
\usepackage{authblk}
\usepackage[table]{xcolor}
\usepackage{hhline}
\usepackage{arydshln}
\usepackage{subfigure}

\usepackage{xcolor}
\newcommand{\blue}[1]{{\leavevmode\color{blue}#1}}
\newcommand{\red}[1]{{\leavevmode\color{red}#1}}

\newcommand{\LA}{\ensuremath{\mathsf{LA}}\xspace}
\newcommand{\SA}{\ensuremath{\mathsf{SA}}\xspace}
\newcommand{\ISA}{\ensuremath{\mathsf{ISA}}\xspace}
\newcommand{\BWT}{\ensuremath{\mathsf{BWT}}\xspace}

\begin{document}
\title{
In-Place BWT and Lyndon Array Construction in Constant Space
}
%
%
\author{Felipe A. Louza\inst{1,2}
\and
Arnaud Lefebvre\inst{3}
}
\authorrunning{F. A. Louza and A. Lefebvre}
%
\institute{Universidade Federal de Uberlândia, MG, Brazil \and
Universidade Estadual de Campinas, SP, Brazil\\
\email{louza@ufu.br}\\
\and
Normandie Univ., UNIROUEN, LITIS, 76000 Rouen, France\\
\email{arnaud.lefebvre@univ-rouen.fr}}
\maketitle              
\begin{abstract}

\noindent
We present an extension of the in-place BWT algorithm of Crochemore et al.~\cite{Crochemore2015} that enables the construction of the Lyndon array using $O(1)$ extra space.  
Our approach incrementally maintains the lexicographic ranks of the
suffixes during the right-to-left BWT construction and then
derives the Lyndon array through a simple next-smaller-value procedure.
Although not intended for practical use due to its quadratic running
time, the method is conceptually simple and works for unbounded alphabets.

\end{abstract}

\section{Introduction}

The Burrows--Wheeler Transform (BWT)~\cite{Burrows1994} and the Lyndon
array (LA)~\cite{Franek2016} are fundamental structures in string processing.  
The BWT underlies compressed indexing~\cite{Ferragina2000} and data compression~\cite{Burrows1994}, while the LA
captures the combinatorial structure of a string through its Lyndon
factorization and plays a role in several pattern-matching~\cite{BannaiKL25,BannaiKKP24} and
combinatorial algorithms~\cite{Bannai2014,BannaiIINTT17}

The BWT depends on the lexicographic ordering of all suffixes of the
input string.  
A standard approach computes it in $O(n)$ time after constructing the
suffix array (SA)~\cite{Manber1993}, which requires $O(n \log n)$ bits of
extra space.  
Okanohara and Sadakane~\cite{Okanohara2009} showed how to compute the BWT
directly in linear time using only $O(n \log \sigma)$ bits of extra
space, where $\sigma$ is the alphabet size, without storing the SA.
Under the most restrictive space bound of $O(1)$ extra words, the
most efficient known construction is the in-place $O(n^2)$-time algorithm of Crochemore et al.~\cite{Crochemore2015}, which overwrites the text from
right to left while maintaining only constant additional memory.  
Their method works in the general ordered alphabet model, requiring only
constant-time symbol comparisons and no assumptions on the alphabet
size.

The Lyndon array can likewise be computed from the suffix ordering in
$O(n)$ time using $O(n \log n)$ bits of extra space
\cite{Franek2016,Franek2017,LouzaMMST19,NUNES2025413}.  
Alternatively, it can be obtained via Duval's factorization
algorithm~\cite{Duval83} in $O(n^2)$ time using $O(1)$ extra space.
More recently, Bille et al.~\cite{BilleE0GKMR20} showed that the Lyndon
array can be computed in $O(n)$ time using only $O(1)$ extra words, and
in $o(n)$ extra bits for the succinct representation~\cite{Louza2018}.
Their algorithm, like that of Crochemore et al., operates in the general
ordered alphabet model and does not depend on $\sigma$.

In this paper we extend the in-place BWT algorithm of Crochemore et
al.~\cite{Crochemore2015} so that it simultaneously
computes the Lyndon array, while still using only $O(1)$ extra space.
The resulting algorithm constructs both structures in quadratic time and
constant space.  
Compared with the linear-time algorithm of Bille et
al.~\cite{BilleE0GKMR20}, our approach is slower, but it has the unique
advantage of producing the BWT \emph{in-place} while computing the Lyndon array simultaneously using the same $O(1)$ extra space.

\section{Preliminaries}\label{sec:prelim}

\begin{figure}[t]
\centering
\setlength{\tabcolsep}{8pt}
\begin{tabular}{c|l|c|c|c|c}
\hline
$i$ & Sorted suffix $T_{\SA[i]}$ & $\SA[i]$ & $\ISA[i]$ & $\BWT[i]$ & $\LA[i]$ \\
\hline
0 & {\$}        & 6 & 4 & {A} & 1 \\
1 & {A\$}       & 5 & 3 & {N} & 2 \\
2 & {ANA\$}     & 3 & 6 & {N} & 1 \\
3 & {ANANA\$}   & 1 & 2 & {B} & 2 \\
4 & {BANANA\$}  & 0 & 5 & {\$} & 1 \\
5 & {NA\$}      & 4 & 1 & {A} & 1 \\
6 & {NANA\$}    & 2 & 0 & {A} & 1 \\
\hline
\end{tabular}
\caption{SA, ISA, BWT and LA for $T=\texttt{BANANA\$}$. 
The Lyndon array $\LA[i]$ gives the length of the longest Lyndon word starting at position $i$.}
\label{fig:matrix}
\end{figure}


Let $\Sigma$ be an ordered alphabet of $\sigma$ symbols, which may be unbounded.
We assume symbols in $\Sigma$ are compared in constant-time.
Let \$ be a special symbol (end marker) not in $\Sigma$ that is lexicographically smaller than every symbol in $\Sigma$.
We denote by $\Sigma^*$ the set of all finite nonempty strings over $\Sigma$, and we define
$\Sigma\$ = \{\,T\$ \mid T \in \Sigma^*\,\}$.

For a string $T = T[0]T[1]\dots T[n-1]$ over $\Sigma\$$, of length $n$, we write $T[i,j] = T[i]T[i+1]\dots T[j]$
for the substring starting at position $i$ and ending at $j$, where $0 \le i \le j < n$.
A \emph{prefix} of $T$ is a substring of the form $T[0,k]$,
and a \emph{suffix} is a substring of the form $T[k,n-1]$.
We denote the suffix that begins at position $i$ by $T_i$.

The \emph{suffix array} (SA)~\cite{Manber1993} of $T\in\Sigma\$$ is an array
$\SA[0,n-1]$ that contains a permutation of $\{0,\dots,n-1\}$ such that
\[
  T_{\SA[0]} < T_{\SA[1]} < \dots < T_{\SA[n-1]}
\]
in lexicographic order.
The \emph{inverse suffix array} (ISA) is the inverse permutation defined by
$\ISA[\SA[i]] = i$.
Hence, $\ISA[k]$ gives the lexicographic rank of the suffix $T_k$.

The \emph{Burrows--Wheeler Transform} (BWT)~\cite{Burrows1994}
of a string $T$ is a permutation of its characters that tends to cluster equal symbols.
It can be defined using the suffix array as
\[
  \BWT[i] =
  \begin{cases}
    T[\SA[i]-1], & \text{if } \SA[i] > 0,\\
    \$, & \text{otherwise.}
  \end{cases}
\]

A \emph{Lyndon word}~\cite{CFL58} is a nonempty string that is strictly smaller in lexicographic order than all of its proper nonempty suffixes.
Every string can be uniquely factorized as a concatenation of Lyndon words
in nonincreasing lexicographic order (the \emph{Lyndon factorization}).

The \emph{Lyndon array} (LA) of a string $T \in \Sigma\$$ is an array $\LA[0,n-1]$ that stores, for each position $i$,
the length of the longest Lyndon word starting at $T[i]$.
Formally,
\[
  \LA[i] = \min\{\,j-i \mid j>i,\ T[i,j-1] \text{ is not smaller than } T[j,n-1]\,\},
\]
or equivalently, the Lyndon array can be computed from the inverse suffix array using the
\emph{next smaller value} (NSV) function:
\[
  \LA[i] = \text{NSV}_{\ISA}(i) - i,
\]
where
\begin{equation}\label{eq:NSVisa}
  \text{NSV}_{\ISA}(i) = \min\{\,j>i \mid \ISA[j] < \ISA[i]\,\},
\end{equation}
and $\text{NSV}_{\ISA}(i)=n$ if no such $j$ exists.

Figure~\ref{fig:matrix} shows the SA, ISA, BWT and Lyndon array for $T=\texttt{BANANA\$}$.

\section{{In-Place BWT}}\label{sec:inplace-bwt}

The in-place BWT algorithm of Crochemore et al.~\cite{Crochemore2015}
computes the transform in quadratic time using only constant extra space, overwriting the input
string from right to left with its BWT as the construction proceeds.

The BWT is computed inductively on the length of the suffix.
Let $BWT(T_s)$ denote the BWT of the suffix $T_s = T[s,n-1]$.
The base cases are the two rightmost suffixes:
\[
BWT(T_{n-2}) = T_{n-2}, \quad BWT(T_{n-1}) = T_{n-1}.
\]
For the inductive step, assume that $BWT(T_{s+1})$ is already stored in $T[s+1,n-1]$.
The algorithm computes $BWT(T_s)$ by inserting the new suffix $T_s$ in its correct
lexicographic position among the suffixes $T_{s+1},\dots,T_{n-1}$.
It relies on the fact that the position of the end marker~\$ in $BWT(T_{s+1})$
corresponds to the rank of the suffix $T_{s+1}$ among the current suffixes.

The algorithm consists of four steps, for iterations $s = n-2, n-3, \dots, 0$.

\begin{enumerate}
\item \textbf{Find the position of the end marker.}
Locate the position $p$ of~\$ in $T[s+1,n-1]$.

\item \textbf{Compute the rank of the new suffix.}
Let $c = T[s]$ be the first symbol of $T_s$.
The local rank $r$ of $T_s$ among the suffixes
$T_s,\dots,T_{n-1}$ is obtained by summing the number of symbols in $T[s+1,n-1]$ that are strictly smaller than $c$, with the number of occurrences of $c$ in $T[s+1,p]$.

\item \textbf{Insert $T_s$ into the sequence.}
Store $c$ into $T[p]$, replacing the end marker~\$.

\item \textbf{Shift symbols to insert a new end marker.}
Shift the block $T[s+1,r]$ one position to the left and write \$ at position $T[r]$.
Then, $BWT(T_s) = T_s$.
\end{enumerate}

Figure~\ref{fig:bwtinplace} illustrates the execution of the
in-place BWT algorithm on the text \texttt{BANANA\$} for the iteration
$s=0$, showing the effect of Steps~1 and~4 on the partial $BWT(T_1)$  stored in $T[1\,..\,n-1]$.

\begin{figure}[t]
\centering
\begin{subfigure}
  \centering
  \begin{tabular}{rc|c|l}
    & s~ & ~\BWT~ & ~sorted suffixes \\
    \hhline{~---}
    $s \rightarrow$ & 0~ & \cellcolor[gray]{0.9}\textcolor{red}{?} &
          ~\cellcolor[gray]{0.9}BANANA\$ \\
    \cdashline{2-4}
    & 1~ & A    & ~\$ \\
    & 2~ & N    & ~A\$ \\
    & 3~ & N    & ~ANA\$ \\
    ~~~~~ $p \rightarrow$ &
      4~ &
      \blue{\$}  &
      ~ANANA\$ \\
    & 5~ &
      A   &
      ~NA\$ \\
    & 6~ & A    & ~NANA\$ \\
    \cline{2-4}
  \end{tabular}
\end{subfigure}
\hfill
\begin{subfigure}
  \centering
  \begin{tabular}{rc|c|l}
    & s~ & ~\BWT~ & ~sorted suffixes \\
    \hhline{~---}
    & 0~ & A    & ~\$ \\
    & 1~ & N    & ~A\$ \\
    & 2~ & N    & ~ANA\$ \\
    & 3~ & \blue{B}  &  ~ANANA\$ \\
    $r \rightarrow$ & 4~ & \cellcolor[gray]{0.9}\textcolor{red}{\$} &
          ~\cellcolor[gray]{0.9}BANANA\$ \\
    & 5~ &
      A   &
      ~NA\$ \\
    & 6~ & A    & ~NANA\$ \\
    \cline{2-4}
  \end{tabular}
\end{subfigure}

\caption{After Step~1 (left) and Step~4 (right) for $T=\texttt{BANANA\$}$, $s=0$.}
\label{fig:bwtinplace}
 \end{figure}

The algorithm runs in $O(n^2)$ time and uses only $O(1)$ additional variables.
Its simplicity and elegant inductive structure make it a useful foundation for
extensions, such as computing the LCP array~\cite{Louza2017b} or, as we show next, the inverse suffix array and then the Lyndon array, while preserving the constant-space property.

\section{{In-place BWT along with the Lyndon Array}}

We now describe how to extend the in-place BWT algorithm (Section~\ref{sec:inplace-bwt}) to compute the BWT and the Lyndon array (LA) simultaneously, using only constant additional memory.

The key idea is to maintain, during each iteration $s=n-2, n-3, \dots, 0$, the lexicographic ranks among the suffixes $T_s,\dots,T_{n-1}$, stored in the $\ISA[s,n-1]$.
At the end, the values in $\ISA[0,n-1]$ are converted into the Lyndon array, according to Equation \ref{eq:NSVisa}.

The algorithm proceeds inductively on the length of the suffix, as in the original in-place BWT construction.
Let $\ISA(T_s)$ denote the partial ISA for $T_s = T[s,n-1]$.
The base cases for ISA are:
\[
\ISA(T_{n-2}) = [1,0], \quad \ISA(T_{n-1}) = [0].
\]
For the inductive step, assume that we also have $\ISA(T_{s+1})$ already stored in $\ISA[s+1,n-1]$.
The algorithm computes $\ISA(T_{s})$ by storing in $\ISA[s]$ the rank $r$ of the new suffix $T_s$ inserted into $BWT(T_{s+1})$, and by updating the ISA-values of all suffix that has rank larger than $r$ in $\ISA[s+1,n-1]$.
Therefore, our algorithm has a new Step 5, added just after Step 4 of the in-place BWT algorithm as follows.

\begin{enumerate}
\item[\textbf{5.}] \textbf{Update the inverse suffix array.}
After inserting $T_s$ in $BWT(T_{s+1})$, set $\ISA[s]=r$, and increment by one all entries $\ISA[i]$, with $s+1 \leq i <n$, such that $\ISA[i] \geq s$.
\end{enumerate}

\begin{lemma}
Step 5 maintains the correct rank information for all suffixes currently present in $BWT(T_s)$.    
\end{lemma}

\begin{proof}
We prove the invariant by induction on $s$ (right to left).  For the
base $s=n-2$ the two ranks are set correctly.  Assume that before
processing $s$ the array $\ISA[s+1..n-1]$ stores the ranks
$0,\dots,m-1$ of $T_{s+1},\dots,T_{n-1}$.  
Let $r\in\{0,\dots,m\}$
be the rank of $T_s$ among $T_s,\dots,T_{n-1}$.  
Every suffix with old rank $q<r$ are shifted one position to the left and keep their rank, while those with $q\ge r$ must have its rank increased by $1$, as $T_s$ was inserted into position $r$; moreover $\ISA[s]$ must be $r$.
Step~5 performs exactly these updates (incrementing entries with
$\ISA[i]\ge r$ and setting $\ISA[s]=r$).  Hence after the update
$\ISA[s..n-1]$ stores the ranks $0,\dots,m$ of $T_s,\dots,T_{n-1}$,
which proves the claim.
\end{proof}

Figure~\ref{fig:bwtisa} illustrates the execution of the
algorithm on the text \texttt{BANANA\$} for the iteration
$s=0$, showing the effect of Steps~1 and~5 on the partial BWT and ISA.

\begin{figure}[t]
\centering
\begin{subfigure}
  \centering
  \begin{tabular}{rc|c|c|l}
    & s~ & ~\BWT~ & ~\ISA~ & ~sorted suffixes \\
    \hhline{~----}
    $s \rightarrow$ & 0~ & \cellcolor[gray]{0.9}\textcolor{red}{?} & \cellcolor[gray]{0.9}\red{?} &
          ~\cellcolor[gray]{0.9}BANANA\$ \\
    \cdashline{2-5}
    & 1~ & A & 3  & ~\$ \\
    & 2~ & N & 5 & ~A\$ \\
    & 3~ & N & 2 & ~ANA\$ \\
    ~~~~~ $p \rightarrow$ &
      4~ & \blue{\$}  & 4 &  ~ANANA\$ \\
    & 5~ & A   & 1 &   ~NA\$ \\
    & 6~ & A   & 0 & ~NANA\$ \\
    \cline{2-5}
  \end{tabular}
\end{subfigure}
\hfill
\begin{subfigure}
  \centering
  \begin{tabular}{rc|c|c|l}
    & s~ & ~\BWT~ & ~\ISA~ & ~sorted suffixes \\
    \hhline{~----}
    & 0~ & A & \red{4}  & ~\$ \\
    & 1~ & N & 3  & ~A\$ \\
    & 2~ & N & \red{6}  & ~ANA\$ \\
    & 3~ & \blue{B}  & 2 &  ~ANANA\$ \\
    $r \rightarrow$ & 4~ & \cellcolor[gray]{0.9}\textcolor{red}{\$} & \cellcolor[gray]{0.9}\textcolor{red}{5} &
          ~\cellcolor[gray]{0.9}BANANA\$ \\
    & 5~ & A  & 1 &   ~NA\$ \\
    & 6~ & A  & 0  & ~NANA\$ \\
    \cline{2-5}
  \end{tabular}
\end{subfigure}

\caption{After Step~1 (left) and Step~5 (right) for $T=\texttt{BANANA\$}$, $s=1$.}
\label{fig:bwtisa}
\end{figure}

{
Finally, given the inverse suffix array $\ISA(T_0)$, it is possible to obtain the Lyndon array, overwriting the value in $\ISA[0,n-1]$ in quadratic time, as follows.
For each position $i=0,1,\dots, n-1$, let
\[
  j = \min\{\,k > i \mid \ISA[k] < \ISA[i]\,\},
\]
or $j = n$ if no such $k$ exists.
Then, $\LA[i] = j - i$.
This corresponds to computing the next smaller value
(NSV) over the ISA, as described in Section~\ref{sec:prelim} (Equation~\ref{eq:NSVisa}).
}
As $\ISA[i]$ is no longer used in the next iteration, we can overwrite it by $\LA[i]$, yielding the final LA without allocating any additional memory.

{
Figure~\ref{fig:isa_la} illustrates how the ISA values maintained
during the in-place BWT construction are subsequently overwritten by the
corresponding Lyndon array values.  
}

In the appendix, we provide a C implementation of the algorithm.

\begin{figure}[t]
    \centering
 \includegraphics[page=1,width=\textwidth]{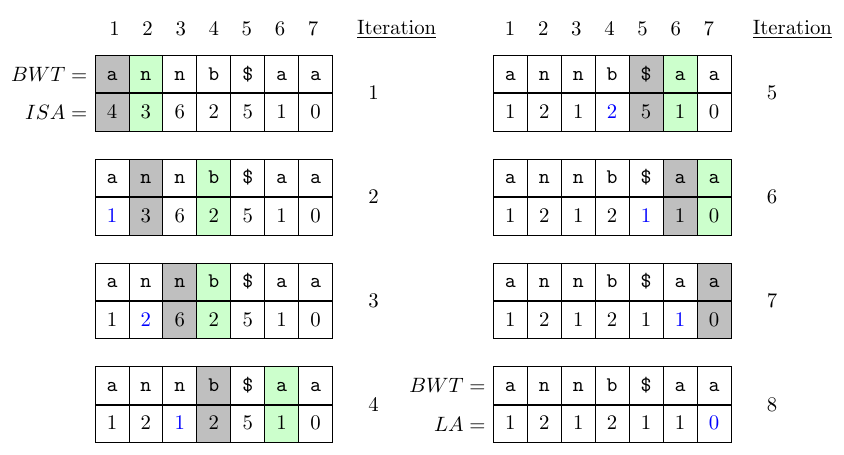}
    \caption{
        Example of the in-place LA construction from ISA 
        on $T=\texttt{BANANA\$}$.
    }
    \label{fig:isa_la}
\end{figure}

\paragraph{\bf Complexity.}
The algorithm follows the same inductive structure as the in-place BWT
construction.  
Each iteration performs a constant amount of work plus two $O(n)$ scans
(Steps~1–2), and there are $n$ iterations, yielding an overall running
time of $O(n^2)$.  
The update of the inverse suffix array in Step~5 requires only an extra
linear scan at each iteration, and the final conversion of $\ISA$ into the Lyndon array by
the next-smaller-value scan also takes $O(n^2)$ time.  
Thus, the total time remains quadratic.
The space usage is $O(1)$ words beyond the input text $T$ and the output
array $\LA$.  As in earlier in-place constructions
\cite{Crochemore2015}, the method is conceptually simple
and demonstrates that both the BWT and the Lyndon array can be computed
together under the strongest possible space constraint.

\section{Conclusion}

We have presented a constant-space algorithm that computes the BWT and the Lyndon array simultaneously.
Our method extends the in-place BWT construction of
Crochemore et al.~\cite{Crochemore2015}
by maintaining the inverse suffix array incrementally and then
deriving the Lyndon array through a simple next-smaller-value
procedure.
The resulting algorithm runs in quadratic time, uses only $O(1)$ additional memory words, and works in the general ordered alphabet model, imposing no bound on the alphabet size.  
As in the original in-place BWT algorithm, the input text is overwritten by its BWT while the Lyndon array is written to a separate array.

Although the quadratic running time limits its practical applicability,
the algorithm shows that the Lyndon array can be computed directly from the text under the
strongest space constraint, even over unbounded alphabets.
Future work includes investigating whether the same strategy can be adapted to compute \emph{succinct} representations of the Lyndon array.


\section*{Funding}
F.A.L acknowledges the financial support from the Brazilian agencies CNPq (grants 408314/2023-0, 311128/2025-4 and 351599/2025-8) and FAPEMIG (grant APQ-01217-22)

\bibliographystyle{plain}
\bibliography{refs}


\newpage
\appendix
\section{Source Code}

\begin{lstlisting}[mathescape=true]
void bwt_lyndon_inplace(unsigned char *T, int n, int* LA){
  int i, p, r=1, s;
  unsigned char c;
  
  // base cases
  LA[n-2] = n-1;
  LA[n-1] = n-2; 
  
  for (s=n-3; s>=0; s--) {
     c = T[s];
  
     /*steps 1 and 2*/
     r = s;
     for (i=s+1, T[i]!=END_MARKER; i++) 
       if(T[i]<=c) r++;
     p = i;
     while (i<n)
       if (T[i++]<c]) r++;
  
     /*step 3*/
     T[p]=c;
  
     /*step 4*/
     for (i=s; i<r; i++) {
       T[i]=T[i+1];
     }
  
     T[r]=END_MARKER;
  
     /*step 5*/
     for (i=s+1; i<n; i++) {
       if(LA[i]<=r) LA[i]--;
     }
     LA[s]= r;
  }
  /*ISA  to LA*/
  for (i=0; i<n; i++) {
    for(j=i+1; j<n; j++){
      if(LA[j]<LA[i])
        break; 
    } 
    LA[i]=j-i;
  }
}

\end{lstlisting}

\end{document}